\documentclass[11pt]{article}
\usepackage{amssymb,amsfonts,amsmath,amsthm}%
\usepackage{graphicx}%
\usepackage{color}%
\usepackage[font=small,format=plain,labelfont=bf,up]{caption}

\usepackage{hyperref}%
\definecolor{gray}{rgb}{0.1,0.1,.1}
\hypersetup{
    colorlinks=true,       % false: boxed links; true: colored links
    linkcolor=gray,          % color of internal links
    citecolor=gray      % color of links to bibliography    
}
\newcommand{\figdraft}{false}%
\newcommand{\figfile}[1]{#1}%
\theoremstyle{plain}%
\newtheorem{theorem}{Theorem}[]%
\newtheorem{corollary}[theorem]{Corollary}%
\newtheorem{lemma}[theorem]{Lemma}%
\newtheorem{remark}[theorem]{Remark}%
\theoremstyle{remark}
\newcommand{\nehari}{\fspace{M}}

\definecolor{colorGreen}{rgb}{0.,0.67,0}
\definecolor{colorRed}{rgb}{0.67,0.,0}
\definecolor{colorBlue}{rgb}{0.,0.,0.67}
\DeclareMathOperator{\argmax}{\mathrm{argmax}}

\DeclareMathOperator{\sgn}{\mathrm{sgn}}

\newcommand{\iu}{\mathtt{i}}
\newcommand{\mhexp}[1]{{{\mathtt{e}}^{#1}}}

\newcommand{\fspace}[1]{{\mathsf{#1}}}
\newcommand{\fspaceL}{\fspace{L}}

\newcommand{\fspaceW}{\fspace{W}}

\newcommand{\phase}{{\varphi}}

\newcommand{\Rset}{{\mathbb{R}}}
\newcommand{\Zset}{{\mathbb{Z}}}

\newcommand{\Nset}{{\mathbb{N}}}
\newcommand{\Qset}{{\mathbb{Q}}}

\newcommand{\ocinterval}[2]{(#1,\,#2]}%
\newcommand{\cointerval}[2]{[#1,\,#2)}%
\newcommand{\oointerval}[2]{(#1,\,#2)}%
\newcommand{\ccinterval}[2]{[#1,\,#2]}%

\newlength{\mhpicDwidth}
\newlength{\mhpicDvsep}
\newlength{\mhpicDhsep}

\newlength{\mhpicPwidth}
\newlength{\mhpicPvsep}
\newlength{\mhpicPhsep}

\newlength{\mhpicWhsep}

\newcommand{\skp}[2]{{\left\langle{#1},\,{#2}\right\rangle}}

\newcommand{\bskp}[2]{{\big\langle{#1},\,{#2}\big\rangle}}

\newcommand{\at}[1]{{\left({#1}\right)}}

\newcommand{\bat}[1]{{\big(#1\big)}}
\newcommand{\Bat}[1]{{\Big(#1\Big)}}

\newcommand{\triple}[3]{{\left({#1},\,{#2},\,{#3}\right)}}

\newcommand{\btriple}[3]{{\big({#1},\,{#2},\,{#3}\big)}}

\newcommand{\bigpar}{\par\quad\newline\noindent}

\newcommand{\norm}[1]{\|{#1}\|}
\newcommand{\abs}[1]{\left|{#1}\right|}
\newcommand{\babs}[1]{\big|{#1}\big|}

\newcommand{\dint}[1]{\,\mathrm{d}#1}

\newcommand{\ga}{{\gamma}}

\newcommand{\la}{{\lambda}}
\newcommand{\si}{{\sigma}}

\newcommand{\calA}{\mathcal{A}}

\newcommand{\calF}{\mathcal{F}}

\newcommand{\calL}{\mathcal{L}}

\newcommand{\calP}{\mathcal{P}}
\newcommand{\calQ}{\mathcal{Q}}

\usepackage[%
a4paper,%
nohead,%
twoside,
ignorefoot,
scale=0.825,
bindingoffset=1.25cm
]{geometry}%
\usepackage{float}%
\begin{document}%
%
%
%
%-------------------------------------------------------------------------------------------
%   Cover page Saarbruecken
%-------------------------------------------------------------------------------------------
%
%\Ptitelseite                 
%
%-------------------------------------------------------------------------------------------
%   Title
%-------------------------------------------------------------------------------------------
%
%
%-------------------------------------------------------------------------------------------
\title{Ground waves in atomic chains \\ with bi-monomial double-well potential}%
\date{\today}%
\author{Michael Herrmann\footnote{Universit\"at des Saarlandes, FR Mathematik, {\tt{michael.herrmann@math.uni-sb.de}}}}
\maketitle
%-------------------------------------------------------------------------------------------
%
%
% \def\theequation{\thesection.\arabic{equation}}
%
%-------------------------------------------------------------------------------------------
%                      abstract
%-------------------------------------------------------------------------------------------
%
\begin{abstract}%
Ground waves in atomic chains are traveling waves that corresponds to minimal non-trivial critical values of the underlying action functional. In this paper we study FPU-type chains with bi-monomial double-well potential and prove the existence of both
periodic and solitary ground waves. To this end we minimize the action on the Nehari manifold and show
that periodic ground waves converge to solitary ones. Finally, we compute
ground waves numerically by a suitable discretization of a constrained gradient flow.
\end{abstract}%
%
%-------------------------------------------------------------------------------------------
%                      MSC and keywords
%-------------------------------------------------------------------------------------------
%
\quad\newline\noindent%
\begin{minipage}[t]{0.15\textwidth}%
Keywords: %
\end{minipage}%
\begin{minipage}[t]{0.8\textwidth}%
\emph{Fermi-Pasta-Ulam chain with double-well potential}, \\
\emph{Nehari manifold}, \emph{ground waves}%
\end{minipage}%
\medskip
\newline\noindent
\begin{minipage}[t]{0.15\textwidth}%
MSC (2010): %
\end{minipage}%
37K60, % 	Lattice dynamics
47J30,  %  Variational methods
74J30  % nonlinear waves
\begin{minipage}[t]{0.8\textwidth}%
\end{minipage}%

%
%
%-------------------------------------------------------------------------------------------
%                      table of contents
%-------------------------------------------------------------------------------------------
%
%
\setcounter{tocdepth}{5} %
\setcounter{secnumdepth}{3}{\scriptsize{\tableofcontents}}%
%
%-------------------------------------------------------------------------------------------
%                      body
%-------------------------------------------------------------------------------------------
%-----------------------------------------------------------------------------------------------
\section{Introduction}
%-----------------------------------------------------------------------------------------------

Atomic chains with nearest neighbor interactions, which are usually called FPU-type chains,  are ubiquitous in physics and materials science as they provide simple atomistic models for solids. They are moreover prototypical examples for nonlinear Hamiltonian lattice equations and shed light on the dynamical properties of discrete media with dispersion.
\par
The lattice equation for infinite FPU-type chains stems from Newton's law of motion and reads
\begin{align}
\label{Eqn:FPU}
\ddot{x}_j=\Phi^\prime\bat{x_{j+1}-x_j}-\Phi^\prime\bat{x_j-x_{j-1}}\,,
\end{align}
where $\Phi$ is the interaction potential and $x_j=x_j\at{t}\in\Rset$ denotes the position of atom $j\in\Zset$ at time $t$. 
\par
Coherent structures such as traveling waves are of particular interest in the analysis of nonlinear lattice equations since they can be regarded as the nonlinear fundamental modes and describe how energy propagates through the chain. Traveling waves are special solutions to \eqref{Eqn:FPU} which depend on a one-dimensional phase variable $\phase=j-\si{t}$ via the ansatz 
\begin{align*}
x_j\at{t}=rj+vt+X\at{j-\si{t}}\,.
\end{align*}
Here, $\si$ the phase speed, $r$ and $v$ are given constants, and the profile function $X$ satisfies
\begin{align}
\label{Eqn:TW1}
\si^2X^{\prime\prime}\at\phase=\Phi^\prime\bat{r+X\at{\phase+1}-X\at\phase}-\Phi^\prime\bat{r+X\at\phase-X\at{\phase-1}}\,.
\end{align}
It can easily be shown that this advance-delay differential equation is equivalent to the nonlinear eigenvalue problem 
\begin{align}
\label{Eqn:TW3}
\si^2W=\calA\Phi^\prime\at{\calA W} + \mu
\end{align}%
for the function $W$ defined by $W\at\phase=r + X^\prime\at{\phase}$. Here, $\mu$ is some constant of integration and the convolution operator $\calA$ is defined by
\begin{align}
\label{Eqn:DefA}
\at{\calA{W}}\at{\phase}=\int_{{\phase}-1/2}^{{\phase}+1/2}W\at{s}\dint{s}\,.
\end{align}
Traveling waves in atomic chains have been studied intensively during the last two decades, and the existence of several types of solutions to \eqref{Eqn:TW1} has been established for various interaction potentials by different methods. For small amplitude waves, the standard references are \cite{FP99} for the continuum approximation
of near-sonic waves and \cite{Ioo00} for bifurcation results via spatial dynamics. The existence of waves with non-small amplitudes has been proven by different variational or critical point techniques, see for instance \cite{FW94,FV99,PP00,SW97,FM02,SZ07,HR10b}, which we discuss below 
in greater detail.
\par
In this paper we consider chains with double-well potential, which play an important role in the atomistic theory of martensitic phase transitions.  In order to keep the presentation as simple as possible, we restrict ourselves to bi-monomial potentials with
\begin{align}
\label{Eqn:Potential}
\Phi\at{w}= d_p\abs{w}^p-d_q\abs{w}^q\quad\text{with} \quad 2<p\,,\quad 1<q<p\,,
\end{align}
and by a simple scaling we can achieve that $d_p=d_q=1$. We further restrict our considerations to the case $\mu=0$, see the discussion below, and seek non-trivial solutions to the traveling wave equation
\begin{align}
\label{Eqn:TW2}
\si^2W+\calA \Psi_{q}\at{\calA W}=\calA\Psi_{p}\at{\calA W}\,,
\end{align}
where the functions $\Psi_r$ are defined by $\Psi_r\at{w}:=r\sgn\at{w}\abs{w}^{r-1}$. 
\paragraph{Variational setting}
Due to the Hamiltonian nature and the shift invariance of \eqref{Eqn:FPU}, there is a variational characterization of traveling waves with prescribed $\si^2$. To see this for both periodic and solitary waves, we introduce 
a positive parameter $K\in\ocinterval{0}{\infty}$. For finite $K$,
we study $2K$-periodic waves and regard
$I_K:=\ocinterval{-K}{K}$
as the periodicity cell; the case $K=\infty$ corresponds to
solitary waves defined on $I_\infty:=\Rset$. In what follows we denote 
by $\fspaceL^r\at{I_K}$ and $\fspaceW^{1,r}\at{I_K}$ the usual Sobolev spaces of $I_K$-periodic functions (or functions on $\Rset$ for $K=\infty$), and write
\begin{align*}
\norm{\cdot}_{r,I_K} \quad\text{and}\quad\skp{\cdot}{\cdot}_{I_K}
\end{align*}
for the norm on $\fspaceL^r\at{I_K}$ and the scalar product in $\fspaceL^2\at{I_K}$, respectively. Moreover, we define
\begin{align*}
\fspace{X}_K:=
\Big\{ W\in\fspaceL^2\at{I_K}\;:\; \calA{W}\in\fspaceL^q\at{I_K}\cap\fspaceL^p\at{I_K}\Big\}
\end{align*}
as ansatz space for periodic and solitary waves.
\par
The variational formulation of \eqref{Eqn:TW2} relies
on the \emph{Lagrangian action functional}
\begin{align*}
\calL_K\at{W}:=\tfrac{1}{2}\si^2\norm{W}_{2,I_K}^2+\calQ_K\at{W}-\calP_K\at{W}\,,
\end{align*}
which is well defined on $\fspace{X}_K$, where
$\tfrac{1}{2}\si^2\norm{W}_{2,I_K}^2$ can be regarded as the kinetic energy, 
and 
\begin{align*}
\calQ_K\at{W}:=\int_{I_K} \abs{\calA W}^q\dint{\phase}\,,\qquad
\calP_K\at{W}:=\int_{I_K} \abs{\calA W}^p\dint{\phase}\,,
\end{align*}
give the two contributions to the potential energy $\calP_K-\calQ_K$. The functionals $\calL_K$,
$\calQ_K$, and $\calP_K$
are also  G\^ateaux-differentiable $\fspace{X}_K$ with derivatives
\begin{align*}
\partial\calL_K\at{W}&=\si^2{W}+\partial\calQ_K\at{W}-\partial\calP_K\at{W}\,,\\
\partial\calQ_K\at{W}&=\calA\Psi_q\at{\calA{W}}\,,\\
\partial\calP_K\at{W}&=\calA\Psi_p\at{\calA{W}}\,,
\end{align*}
and we conclude that each solution $W\in\fspace{X}_K$
to the traveling wave equation \eqref{Eqn:TW2} is a critical point of $\calL_K$, and vice versa.
\paragraph{Main result and organisation of paper}
In this paper we establish the existence of periodic and solitary \emph{ground waves} for
the special potential \eqref{Eqn:Potential}. By definition, a ground wave is a traveling wave that corresponds to a minimal non-trivial critical value of the action functional $\calL_K$. Our main result can be summarized as follows.
\begin{theorem}
\label{Thm:MainResult}
Let $\si^2>0$ be given. Then, for each $K\in\ocinterval{0}{\infty}$ there exists a ground wave $W_{K}\in\fspace{X}_K$ which satisfies \eqref{Eqn:TW2} as well as
\begin{align*}
0<\calL_K\at{W_{K}}=\min\Big\{\calL_K\at{W}\;:\;W\in\fspace{X}_K\setminus\{0\}\text{ with  }\partial\calL_K\at{W}=0\Big\}\,,
\end{align*}
and which is non-constant provided that $K$ is sufficiently large. Moreover, 
solitary ground waves can be approximated (in some strong sense) by period ground waves.
\end{theorem}
\noindent
Before commenting on the main ideas in the proof, we proceed with the following
remarks:
\begin{enumerate}
\item
Our proof reveals that each ground waves $W_K$ satisfies
\begin{align*}
\calL_K\at{W_K}=\min_{W\in\fspace{X}_K\setminus\{0\}}\max_{\zeta>0}\calL_K\at{\zeta W}\,,
\end{align*}
that means ground waves are not minimizers but saddle points of $\calL_K$.
\item The traveling wave equation \eqref{Eqn:TW2} and the properties of the convolution operator
$\calA$, see Lemma \ref{Lem:AProps} below, ensure that travelling waves belong to
$\fspace{BC}^1\at{\Rset}$ and that solitary ground waves are
homoclinic via $\lim_{\phase\to\pm\infty}W\at{\phase}=0$.
Moreover, each ground wave is supersonic in the sense of $\si^2>0\geq\Phi^{\prime\prime}\at{0}$.

\item
The set of all ground waves for given $\si^2>0$ is obviously invariant under shifts $W\at\phase\rightsquigarrow{W}\at{\phase+\phase_0}$, reflections 
$W\at\phase\rightsquigarrow{W}\at{-\phase}$, and sign changes $W\at\phase\rightsquigarrow{-W\at\phase}$. 
\item 
Since $\Phi$ is a double-well potential, there exists an \emph{amplitude threshold}, this means a lower bound for 
the strain amplitude $\norm{AW}_{\infty,I_K}$ independent of $\si$ and $K$, see
Remark \ref{Rem:AmpThreshold} below. This is different from the case of convex potentials as these support
near sonic waves with arbitrarily small amplitudes \cite{FP99,Her10a}.
\item 
Numerical simulations as presented in \S\ref{sec:num} indicate that the solitary ground waves
provided by Theorem \ref{Thm:MainResult} have, for sufficiently small $\si$,
oscillatory tails. This is also different from the case of convex potentials, where 
the profile function $W$ of solitary waves usually has only one local extremum \cite{Her10a}.%
\item 
Although our existence result is restricted to bi-monomial double-well potentials, we expect that
the underlying ideas can -- for the price of more technical effort -- be generalized
to more general double-well potentials $\Phi\at{w}=-\Phi_1\at{w}+\Phi_2\at{w}$ as long as
$\at{i}$ the partial potentials $\Phi_1$ and $\Phi_2$ are convex, $\at{ii}$ 
$\Phi_1$ dominates $\Phi_2$ for all sufficiently small  $\abs{w}$, and $\at{iii}$
$\Phi_2$ grows superquadratically and, at least for large $\abs{w}$, faster than 
$\Phi_1$.
\end{enumerate}
We also mention that $\calL_K$ has, at least for $K<\infty$, 
infinitely many critical points in $\fspace{X}_K$. In particular, each ground wave for $K/n$ with $n\in\Nset$ is also a critical point of $\calL_K$, but qualitatively different types of traveling waves might exist as well. Moreover, 
for $K=\infty$ we expect that there also exist traveling waves with non-decaying profile $W$. Of particular importance are phase transitions waves, which are heteroclinic connections of two period waves corresponding to either one of wells. Unfortunately, very little is known about their existence. The only available results concern bi-quadratic potentials, which allow for solving \eqref{Eqn:TW1} by Fourier methods \cite{TV05,SZ09}, or almost bi-quadratic potentials, for which we
can employ perturbation methods \cite{HMSZ12}. It remains a challenging task to find alternative, maybe variational, existence proofs for phase transition waves that apply to more general double-well potentials.
\bigpar
In order to prove Theorem \ref{Thm:MainResult}, we introduce in \S\ref{sec:setting} the Nehari manifold $\nehari_K$, which has co-dimension $1$ and contains all non-trivial critical points of $\calL_K$. In \S\ref{sec:min.per} we employ the direct method from the calculus of variations and show that the functional $\calL_K$ attains its minimum on $\nehari_K$  for $K<\infty$. Afterwards in \S\ref{sec:min.sol} we demonstrate that periodic ground waves with $K\to\infty$ provide minimizing sequences for $\calL_\infty|_{\nehari_\infty}$ and converge to a solitary ground wave. Finally, we compute ground waves numerically in \S\ref{sec:num}.
\paragraph{Comparison with other results}
We finally compare our results and methods with previous work on traveling waves in
FPU-type chains.
\par
At first we mention that the amplitude threshold implies that 
traveling waves in double-well potentials cannot be obtained by
perturbation arguments applied to the trivial wave $W\at\phase\equiv0$. Moreover, also the
available existence proofs via constrained optimization does not cover 
potential~\eqref{Eqn:Potential}.  Specifically,
both the method used by Friesecke and Wattis in 
\cite{FW94} (minimization of kinetic energy under prescribed potential energy)
and the approach discussed in \cite{Her10a} (maximization of potential energy under prescribed kinetic energy) require --
among other conditions --
that the asymptotic state of the solitary wave $w_*=\lim_{\phase\to\pm}W\at\phase$ is a minimum
of the tilted potential 
\begin{align}
\label{Eqn:TiltedPot}
\Phi_*\at{w}=\Phi\at{w_*+w}-\Phi^\prime\at{w_*}{w}-\Phi\at{w_*}\,,
\end{align}
and this condition is apparently not satisfied for the waves provided by Theorem \ref{Thm:MainResult}
(notice that in our case we have $w_*=0$ and hence $\Phi_*=\Phi$). 
\par
Theorem \ref{Thm:MainResult} is, however, not the first existence result for
traveling waves in bi-monomial double-well potentials. Smets and Willem
\cite{SW97} combine the Mountain Pass Theorem with weak convergence methods
to establish the existence of solitary waves for a large class of potentials including \eqref{Eqn:Potential}.
However, these waves are not necessarily ground waves and it is not clear
whether they can be approximated by periodic ones. A further drawback is that 
solutions provided by mountain pass arguments are hard to compute numerically.
\par
Closely related to our work is the discussion of ground waves given by Pankov in \cite[Section 3.4]{Pan05}. The results presented there imply the assertions of Theorem \ref{Thm:MainResult} for the special case $q=2$ and are likewise based on the Nehari manifold and approximation with periodic waves. The proof, however, is different as it employs the Mountain Pass Theorem and the Palais-Smale condition for $K<\infty$; see \S\ref{sec:min} for more details.
\par
A variant of the Mountain Pass Theorem was also used by Schwetlick and Zimmer \cite{SZ07} to construct homoclinic waves for certain double-well potentials. These waves satisfy $W=w_*+\tilde{W}$ with $\lim_{\abs\phase\to\infty} \tilde{W}\at\phase=0$, where $w_*$ is one of the local minimizer of $\Phi$. The key idea is that the relative profile $\tilde{W}$ can be regarded as solitary wave with respect to the tilted potential \eqref{Eqn:TiltedPot}, which has, at least for certain double-well potentials, sufficiently nice properties. The relation to our approach becomes apparent in the periodic case. Instead of tilting the potential we can impose the constraint $\abs{I_K}^{-1}\int_{I_K}W\dint\phase=w_*$, and we easily check that critical points of $\calL_K$ now satisfy the traveling wave equation \eqref{Eqn:TW3} with Lagrangian multiplier $\mu\in\Rset$. Due to the constraint, however, the corresponding traveling waves
are not ground waves for the action $\calL_K$.
\par
%
%
% 

%-----------------------------------------------------------------------------------------------
\section{Preliminaries and Nehari manifold}\label{sec:setting}
%-----------------------------------------------------------------------------------------------
%
In this section we develop our variational framework for both finite and infinite $K$ and introduce the Nehari manifold, on which we minimize the action in \S\ref{sec:min}. The parameter $\si^2>0$ is from now on fixed.
\par
We first summarize some properties of the convolution operator $\calA$. In particular, we show 
that $\calA$ maps $\fspaceL^2\at{I_K}$ compactly into $\fspaceL^r\at{I_K}$ provided that $1\leq{r}<\infty$ and $K<\infty$.
\begin{lemma}%[Properties of $\calA$]
\label{Lem:AProps}
Let $K\in\ocinterval{0}{\infty}$ and $1\leq{r}<\infty$ be given. Then, the linear operator $\calA$ maps $\fspaceL^{r}\at{I_K}$ continuously into $\fspaceW^{1,r}\at{I_K}\subset\fspace{BC}\at\Rset$ with
\begin{align*}
\norm{\at{\calA{W}}^\prime}_{r,I_K}\leq 2\norm{W}_{r,I_K}\,,\qquad
\norm{\calA{W}}_{\infty,I_K}\leq \norm{W}_{r,I_K}\,,\qquad
\norm{\calA{W}}_{r,I_K}\leq \norm{W}_{r,I_K}\,.
\end{align*}
Moreover, $W_n\rightharpoonup{W}_\infty$ weakly in $\fspaceL^2\at{I_K}$ implies $\calA{W_n}\to\calA{W}_\infty$ pointwise for all $K$ and also strongly in $\fspaceL^r\at{I_K}$ for $K<\infty$.
\end{lemma}
\begin{proof}
Thanks to $\at{\calA{W}}^\prime\at\phase=W\at{\phase+1/2}-W\at{\phase-1/2}$ and since H\"older's inequality implies
\begin{align*}
\babs{\at{\calA{W}}\at{\phase}}^r\leq\int_{\phase-1/2}^{\phase+1/2}\abs{W\at{s}}^r\dint{s}\,,
\end{align*}
all estimates follows immediately. Moreover, the pointwise convergence $\calA{W_n}\to\calA{W}$ follows from the definition of $\calA$ in \eqref{Eqn:DefA}, and implies the strong convergence for $K<\infty$ due to  $\fspaceL^\infty\at{I_K}\subset\fspaceL^r\at{I_K}$ and  the Dominated Convergence Theorem.
\end{proof}
\begin{remark}
\label{Rem:KernelA}%
For $K=\infty$ or $K\notin\pi\Qset$, the operator $\calA:\fspaceL^2\at{I_K}\to\fspaceL^2\at{I_K}$ has only trivial kernel.
\end{remark}
\begin{proof} The operator $\calA$ diagonalizes in Fourier space via 
$\calA\mhexp{\iu k\phase}=\varrho\at{k/2}\mhexp{\iu{k}\phase}$
for all $k\in\Rset$, where
$\varrho\at{\kappa}=\sin\at{\kappa}/\kappa$, and the
assertion follows immediately.
\end{proof}
Notice that Lemma \ref{Lem:AProps} implies
$\fspace{X}_K=\fspaceL^2\at{I_K}$ for $K<\infty$,
as well as
$\calP_K\at{W}\leq\norm{W}_{2,I_K}^{p-q}\calQ_K\at{W}$ for 
all $K\in\ocinterval{0}{\infty}$ and $W\in\fspace{X}_K$.
%-----------------------------------------------------------------------------------------------
\subsection{Necessary condition for traveling waves}
%-----------------------------------------------------------------------------------------------
%
%
A key property of the action functional $\calL_K$ is that its restriction 
to the positive ray $\zeta>0\mapsto\zeta{W}$ has a unique maximizer for every non-degenerate $W$. To see this, we
start with an auxiliary result about the maximizers of certain tri-monomial functions. 
\begin{lemma}
\label{Lem:AuxRes}
For any $c=\triple{c_2}{c_q}{c_p}\in\Rset_+^3$ let the function
$\la_c:\Rset_+\to\Rset$ be defined by
\begin{align*}
\la_c\at{\xi}:=c_2\xi^2+c_q\xi^q-c_p\xi^p\,.
\end{align*}
Then, the function $\bar\xi\Rset_+^3\to\Rset$ with
\begin{align*}
\bar{\xi}\at{c}:=\argmax\limits_{\xi>0} \la_c\at{\xi}
\end{align*}
is well-defined, continuous, and  positive.
\end{lemma}
\begin{proof}
The function $f_c$ with  
\begin{align*}
f_c\at\xi:=\xi \la_c^\prime\at\xi=2c_2\xi^2+qc_q\xi^q-pc_p\xi^p\,,
\end{align*}
is continuous on $R_+$ and satisfies
\begin{align*}
f_c\at{\xi}>0\qquad\text{and}\qquad
f_c\at{\xi}<0
\end{align*}
for small and large $\xi$, respectively. Therefore, 
there exists at least one zero $\bar\xi>0$ with
\begin{align}
\label{Lem:AuxRes.Eqn1}
0=f_c\at{\bar\xi}=
2c_2\bar{\xi}^2+qc_q\bar{\xi}^q-pc_p\bar{\xi}^p\,.
\end{align}
Thanks to \eqref{Lem:AuxRes.Eqn1} we also find 
\begin{align*}
\bar\xi f_c^\prime\bat{\bar\xi}=\at{4-2p}c_2\bar{\xi}^2+\at{q^2-qp}c_q\bar{\xi}^q<0\,,
\end{align*}
for any such zero $\bar\xi$, and this implies that
$f_c$ has precisely one zero $\bar\xi=\bar\xi\at{c}$. In particular, we
have
\begin{align*}
\la_c^\prime\at{\xi}>0\qquad\text{and}\qquad
\la_c^\prime\at{\xi}<0
\end{align*}
for $\xi<\bar\xi\at{c}$ and $\xi>\bar\xi\at{c}$, respectively,  and we conclude that $\bar\xi\at{c}$ is a global maximizer of $\la_c$
with 
\begin{align*}
0<\la_c\bat{\bar\xi\at{c}}=\max\limits_{\xi>0} \la_c\at{\xi}\,.
\end{align*}
Finally, using \eqref{Lem:AuxRes.Eqn1} once more we verify the estimates
\begin{align*}
\max \left\{
\at{\frac{2c_2}{pc_p}}^{1/\at{p-2}},\;
\at{\frac{qc_q}{pc_p}}^{1/\at{p-q}}
\right\}
\leq\bar\xi\at{c}\leq%
\max \left\{
\at{\frac{4c_2}{pc_p}}^{1/\at{p-2}},\;
\at{\frac{2qc_q}{pc_p}}^{1/\at{p-q}}
\right\},
\end{align*}
which in turn imply the claimed continuity of $\bar\xi$ since we have $\la_{c_n}\to{\la}_c$ uniformly 
as $c_n\to{c}\in\Rset_+^3$ on each compact subset of $\Rset_+$.
\end{proof}
\begin{corollary}%[Unique maximizer along each ray]
\label{Cor:MaxAlongRays}
For $K\in\ocinterval{0}{\infty}$ and $W\in\fspace{X}_K$ with $\calA{W}\neq0$ let
\begin{align*}
\bar{\zeta}_K\at{W}:=
\bar{\xi}\btriple{\tfrac{1}{2}\si^2\norm{W}_{2,I_K}^2}{\calQ_K\at{W}}{\calP_K\at{W}}\,.
\end{align*}
Then, we have
\begin{align*}
0<\calL_K\at{\bar\zeta_K\at{W}W}=\max_{\zeta>0}\calL_K\at{\zeta{W}}\,
\end{align*}
and $W_n\to W$ strongly in $\fspaceL^2\at{I_K}$ with $\calA W\neq0$ implies
$\bar\zeta_K\at{W_n}\to\bar{\zeta}_K\at{W}$.
\end{corollary}
\bigpar
In view of Corollary \ref{Cor:MaxAlongRays}, we introduce the functional
\begin{align*}
\calF_K\at{W}:=\frac{\dint}{\dint\zeta}\calL_K\at{\zeta{W}}|_{\zeta=1}=\bskp{\partial\calL_K\at{W}}{W}_{I_K}=\si^2\norm{W}_{2,I_K}^2+q\calQ_K\at{W}-p\calP_K\at{W}\,,
\end{align*}
which is well defined and  G\^ateaux differentiable on $\fspaceL^2\at{I_K}$ with
\begin{align*}
\partial\calF_K\at{W}=2\si^2{W}+q\calA\Psi_q\at{\calA{W}}-p\calA\Psi_p\at{\calA{W}}\,.
\end{align*}
We further define the \emph{Nehari manifold} by
\begin{align*}
\nehari_K&:=\Big\{W\in\fspace{X}_K\;:\;W\neq{0},\quad\calF_K\at{W}=0\Big\}\,,
\end{align*}
and notice that  $\bar\zeta_K\at{W}{W}\in\nehari_K$ for all $W$ with $\calA{W}\neq{0}$. Moreover, 
$W\in\nehari_K$ implies $\calA{W}\neq0$, and for $\calA{W}\neq0$ we have
\begin{align}
\notag%\label{Eqn:NehariProps}
W\in\nehari_K\quad\Longleftrightarrow\quad \bar{\zeta}_K\at{W}=1
\quad\Longleftrightarrow\quad
\calL_K\at{W}=\max_{\zeta>0}\calL_K\at{\zeta{W}}\,.
\end{align}
\begin{remark}
\label{Rem:Nehari} %
Each non-vanishing traveling wave $W\in\fspace{X}_K$ belongs to $\nehari_K$ and $\fspace{BC}^1\at\Rset$.
\end{remark}
\begin{proof}
$W\neq0$ combined with the traveling wave equation \eqref{Eqn:TW2} implies $\calA{W}\neq0$ thanks to $\si\neq0$, and 
$W\in\nehari_K$ follows since testing \eqref{Eqn:TW2} with $W$ gives $\calF_K\at{W}=0$. Moreover,
$W\in\fspace{BC}^1\at\Rset$ is a direct consequence of Lemma \ref{Lem:AProps} and \eqref{Eqn:TW2}.
\end{proof}
Our strategy for proving the existence of ground waves is to show that
$\calL_K$ attains its minimum on $\nehari_K$. We can then conclude that   
each minimizer satisfies the traveling wave equation \eqref{Eqn:TW2}, and Remark 
\ref{Rem:Nehari} guarantees that the minimum is in fact the smallest non-trivial critical value of $\calL_K$.
%
%
%-----------------------------------------------------------------------------------------------
\subsection{Properties of the Nehari manifold}
%-----------------------------------------------------------------------------------------------
%
%
We next derive some estimates on the Nehari manifold.
\begin{lemma}%[estimates on Nehari manifold]
\label{Lem:Nehari.Estimates}
There exist positive constants $c$ and $C$, which both are independent of $K$ but can depend on $\si$, such that 
\begin{enumerate}
\item
\begin{math}
c\leq \norm{W}^2_{2,I_K}\leq{C}\calL_K\at{W},
\end{math}
\item
\begin{math}
c\leq\norm{\calA{W}}^2_{\infty,I_K}\leq{C}\calL_K\at{W},
\end{math}
\item 
\begin{math}
c\leq\calL_K\at{W},
\end{math}
\item 
\begin{math}
c\leq\calP_K\at{W}\leq{C}\calL_K\at{W},
\end{math}
\item 
\begin{math}
\calQ_K\at{W}\leq{C}\calL_K\at{W},
\end{math}
\item
\begin{math}
\bskp{\partial{F}_K\at{W}}{W}_{I_K}\leq-{c},
\end{math}

\end{enumerate}
hold for all $W\in\nehari_K$.
\end{lemma}
\begin{proof}
Employing $\calF_K\at{W}=0$ and Lemma \ref{Lem:AProps} we estimate
\begin{align*}
q\calQ_K\at{W}\leq
\si^2\norm{W}_{2,I_K}^2+q\calQ_K\at{W}=p\calP_K\at{W}
\leq{p}\norm{\calA{W}}_{\infty,I_K}^{p-q}\calQ_K\at{W}\,,
\end{align*}
where $W\in\nehari_K$ implies $AW\neq0$ and hence $\calQ_K\at{W}>0$.
Therefore, and thanks to Lemma \ref{Lem:AProps}, we find
\begin{align}
\label{Lem:Nehari.Estimates.Eqn1}
\norm{W}_{2,I_K}\geq\norm{\calA{W}}_{\infty,I_K}\geq\at{\frac{q}{p}}^{1/\at{p-q}}\,,
\end{align}
and this implies
\begin{align}
\label{Lem:Nehari.Estimates.Eqn2}
p\calP_K\at{W}=
\si^2\norm{W}_{2,I_K}^2+q\calQ_K\at{W}
\geq
\si^2\norm{W}_{2,I_K}^2\geq \si^2\at{\frac{q}{p}}^{2/\at{p-q}}\,.
\end{align}
Due to $\calF_K\at{W}=0$ we also have
\begin{align*}
\calL_K\at{W}=\si^2\at{\frac{1}{2}-\frac{1}{p}}\norm{W}_{2,I_K}^2+\at{1-\frac{q}{p}}\calQ_K\at{W}
\geq
\si^2\at{\frac{1}{2}-\frac{1}{p}}\norm{W}_{2,I_K}^2,
\end{align*}
and combining this estimate with \eqref{Lem:Nehari.Estimates.Eqn1}, 
\eqref{Lem:Nehari.Estimates.Eqn2}, and $\calF_K\at{W}=0$ we arrive at the first five assertions.
Moreover, a direct computation yields
\begin{align*}
\bskp{\partial{F}_{K}\at{W}}{W}_{I_K}&=2\si^2\norm{W}_{2,I_K}^2+q^2\calQ_K\at{W}-p^2\calP_K\at{W}\\
&=\at{2-p}\si^2\norm{W}_{2,I_K}^2+\at{q^2-qp}\calQ_K\at{W}\leq-\abs{p-2}\si^2\norm{W}_{2,I_K}^2,
\end{align*}
which in turn implies the sixth assertion.
\end{proof}
The estimate \eqref{Lem:Nehari.Estimates.Eqn1} reveals 
that the lower bound for the strain amplitude in Lemma \ref{Lem:Nehari.Estimates} is actually independent of $\si$,
and this implies the amplitude threshold for traveling waves in bi-monomial double-well potentials. 
\begin{remark}
\label{Rem:AmpThreshold}
There exists a constant $c$ independent of both $K$ and $\si$
such that
\begin{align*}
\norm{AW}_{\infty,I_K}\geq{c},
\end{align*}
holds for any non-trivial traveling wave $W$.
\end{remark}
The last assertion in Lemma  \ref{Lem:Nehari.Estimates} ensures that the Nehari manifold $\nehari_K$ is strictly transversal to each positive ray $\zeta>0\mapsto\zeta{W}$ with $W\in\nehari_K$. It is therefore clear that minimizers of $\calL_K|_{M_K}$ are critical points of $\calL_K$.  Here we give an alternative proof of this assertion that relies
on the constrained gradient flow for $\calL_K$, that is 
\begin{align}
\label{Eqn:GradientFlow}
\tfrac{\dint}{\dint\tau}{W_\tau}=-\partial\calL_K\at{W_\tau}+\la_K\at{W_\tau}\partial{F}_K\at{W_\tau}\,,\qquad
\la_K\at{W}:=
\frac{\bskp{\partial\calL_K\at{W}}{\partial\calF_K\at{W}}_{I_K}}{\norm{\partial\calF_K\at{W}}_{2,I_K}^2}\,,
\end{align}
where $\tau$ is the flow time and $\tau\mapsto{W_\tau}$ denotes a curve in $\nehari_K$. This gradient flow 
is also the starting point for the numerical approximation of ground waves in \S\ref{sec:num}.
\begin{lemma}%[Properties of the constrained gradient flow] 
\label{Lem:GradFlowProps}
For each $K\in\ocinterval{0}{\infty}$, the initial value problem to the $\nehari_K$-valued ODE \eqref{Eqn:GradientFlow} is well-posed. Moreover, $\calL_K$ is strictly decreasing along each non-stationary trajectory and each stationary point solves the traveling wave equation \eqref{Eqn:TW2}.
\end{lemma}
\begin{proof} 
Let some initial datum $W_0\in\nehari_K$ be given. Lemma \ref{Lem:Nehari.Estimates} implies $\partial\calF_K\at{W_0}\neq0$, and by continuity there exists a small ball $B\subset\fspace{X}_K$ around $W_0$ such that the multiplier $\la_K$ is a well defined and Lipschitz continuous function on $B$. Consequently, 
there exists a local solution $\tau\in\cointerval{0}{\tau_1}\mapsto{W_\tau}\in\fspace{X}_K$. The definition of $\la_K$ implies
\begin{align*}
\tfrac{\dint}{\dint\tau}\calF_K\at{W_\tau}&=\bskp{\partial\calF_K\at{W_\tau}}{\tfrac{\dint}{\dint\tau}W_\tau}_K=0\,,
\end{align*}
so $\nehari_K$ is indeed invariant under the flow of \eqref{Eqn:GradientFlow}. Moreover, a direct computation gives
\begin{align*}
\tfrac{\dint}{\dint\tau}\calL_K\at{W_\tau}&=\bskp{\partial\calL_K\at{W_\tau}}{\tfrac{\dint}{\dint\tau}W_\tau}_{I_K}
\\&=%
\frac{\norm{\partial\calL_K\at{W_\tau}}_{2,I_K}^2
\norm{\partial\calF_K\at{W_\tau}}_{2,I_K}^2-{\bskp{\partial\calL\at{W_\tau}}{\partial\calF\at{W_\tau}}_{I_K}}^2}{\norm{\partial\calF_K\at{W_\tau}}_{2,I_K}^2}
\\&\geq0\,,
\end{align*}
where the inequality is strict provided that $\partial\calL_K\at{W_\tau}$ and $\partial\calF_K\at{W_\tau}$ are not co-linear, that means as long as the right hand side in \eqref{Eqn:GradientFlow}$_1$ does not vanish. Finally, suppose that $W_\tau\equiv{W}\in\nehari_K$ is stationary under the flow of \eqref{Eqn:GradientFlow}. Then we have
\begin{align*}
\partial\calL_K\at{W}=\la_K\at{W}\partial\calF_K\at{W}\,,
\end{align*}
and testing this identity with $W$ we find
\begin{align*}
0=\bskp{\partial\calL_K\at{W}}{W}_{I_K}=\la_K\at{W}\bskp{\partial\calF_K\at{W}}{W}_{I_K}\,.
\end{align*}
Lemma \ref{Lem:Nehari.Estimates} now implies $\la_K\at{W}=0$, and hence $\partial\calL_K\at{W}=0$.
\end{proof}
A particular consequence of Lemma \ref{Lem:GradFlowProps} is that each minimizer of $\calL_K|_{M_K}$ is a stationary point of \eqref{Eqn:GradientFlow}, and thus in fact a traveling wave.
%
%-----------------------------------------------------------------------------------------------
\section{Ground waves as Nehari minimizers of the action}\label{sec:min}
%-----------------------------------------------------------------------------------------------
%
In this section we finish the proof of Theorem \ref{Thm:MainResult} by showing that $\calL_K$ attains its minimum on the Nehari manifold $\nehari_K$. To this end we employ the direct method for $K<\infty$, and pass afterwards to the limit $K\to\infty$. For the proofs we define
\begin{align*}
\ell_K:=\inf \calL_K|_{\nehari_K}
\end{align*}
and recall that Lemma \ref{Lem:Nehari.Estimates} provides a constant $c>0$ such that
$\ell_K\geq{c}$ for all $K\in\ocinterval{0}{\infty}$.
%
%
%
%-----------------------------------------------------------------------------------------------
\subsection{Existence of periodic ground waves}\label{sec:min.per}
%-----------------------------------------------------------------------------------------------
%
We now fix $0<K<\infty$ and employ the compactness of $\calA$ to show that each minimizing sequence for $\calL_K|_{\nehari_K}$ contains a subsequence that converges to a minimizer. Alternatively, we could employ critical point techniques as follows.  Using similar estimates as in the proof of Lemma \ref{Lem:Nehari.Estimates}, one easily shows that
the action landscape has a \emph{mountain pass geometry} via
\begin{align*}
\calL_K\at{0}=0\,,\qquad \inf\limits_{\norm{W}_{2,I_K}=1}\calL_K\at{W}\geq\tfrac{1}{2}\si^2\,,\qquad
\lim_{\zeta\to\infty} \calL_K\at{\zeta{W}}=-\infty\;\;\;\text{for}\;\;\;\calA{W}\neq0\,,
\end{align*}
and the compactness of $\calA$ ensures that $\calL_K$ satisfies the Palais-Smale condition.
The existence of non-vanishing critical values is hence implied by the Mountain Pass Theorem, and the
Palais-Smale condition guarantees that there is a minimal critical value. The details for this line of argument can, for the special case $q=2$, be found in \cite[Section 3.4.1]{Pan05}.
\begin{theorem}
\label{Thm:ExistencePeriodicWaves}
For each $0<K<\infty$ there exists a minimizer of $\calL_K|_{\nehari_K}$.
\end{theorem}
\begin{proof}
Let $\at{W_n}_{n\in\Nset}\subset\nehari_K$ be any minimizing sequence for $\calL_K|_{\nehari_K}$. By construction and  Lemma \ref{Lem:Nehari.Estimates}, we then have
\begin{align*}
c\leq\norm{W_n}_{2,I_K}\leq{C}\,,\qquad
c\leq\calL_K\at{W_n}\leq{C}\,,\qquad
c\leq\calP_K\at{W_n}\leq{C}
\end{align*}
for some constants $0<c<C<\infty$ independent of $n$. By passing to a (not relabeled) subsequence we can assume that
$W_n\rightharpoonup{W}$ weakly in $\fspaceL^2\at{I_K}$, and that $\lim_{n\to\infty}\norm{W_n}_{2,I_K}^2$ exists, and this implies
\begin{align*}
\ga:=\lim_{n\to\infty}\norm{W_n}_{2,I_K}-\norm{W}_{2,I_K}\geq0\,. 
\end{align*}
Our strategy is now to show that $W$ minimizes $\calL_K$ on $\nehari_K$. The properties of $\calA$, see Lemma \ref{Lem:AProps}, guarantee that $\calA{W_n}\to\calA{W}$ pointwise and strongly in $\fspaceL^{r}\at{I_K}$ for all $1<r<\infty$. We therefore have $W\in\fspace{X}_K$, and
$\calP_K\at{W}=\lim_{n\to\infty}\calP_K\at{W_n}\geq{c}$ yields
$\calA{W}\neq0$ and hence $W\neq0$.
From the definition of $\nehari_K$ we infer that
\begin{align*}
\ell_K&=\calL_K\at{W_n}\geq\calL_K\at{\zeta W_n}
\\&\geq 
\frac{1}{2}\zeta^2\si^2\at{\lim_{n\to\infty}\norm{W_n}_{2,I_K}^2-\norm{W}_{2,I_K}^2}+
\frac{1}{2}\si^2\norm{\zeta W}^2_{2,I_K}+\calQ_k\at{\zeta{W_n}}-\calP_K\at{\zeta{W_n}}
\end{align*}
for all $\zeta>0$, and the above convergencies imply 
\begin{align*}
\calL_K\at{\zeta W}\leq\frac{1}{2}\zeta^2\si^2\ga+\ell_K\,.
\end{align*}
Choosing $\zeta=\bar\zeta_K\at{w}\in\oointerval{0}{\infty}$
gives
\begin{align*}
\ell_K\leq\calL_K\at{\bar\zeta_K\at{W}{W}}\leq\frac{1}{2}\si^2\bar\zeta_K\at{W}^2\ga+\ell_K\,,
\end{align*}
and we conclude that
\begin{align*}
\ell_K=\max_{\zeta>0}\calL_K\at{\zeta W},\qquad\norm{W}_{2,I_K}=\lim_{n\to\infty}\norm{W_n}_{2,I_K}\,.
\end{align*}
These identities imply the strong convergence $W_n\to{W}$ as well as $W\in\nehari_K$ with 
$\calL_K\at{W}=\ell_K=\lim_{n\to\infty}\calL_K\at{W_N}$. 
\end{proof}
To complete our existence proof for periodic ground waves we finally derive an upper bound for $\ell_K$ which in turn implies
that minimizer of $\calL_K$ are non-trivial for large $K$.
\begin{lemma}
\label{Lem:LimSupEst}
We have $\limsup_{K\to\infty}\ell_K\leq\ell_\infty$.
\end{lemma}
\begin{proof} 
Let $W_\infty\in\nehari_\infty$ be given. For each $1<K<\infty$ we define $V_K\in\fspace{X}_\infty$
\begin{align*}
V_K
\at{\phase}:=\left\{\begin{array}{lcl}
W_\infty\at{\phase}&&\text{for $\abs{\phase}<K-\tfrac12$},
\\0&&\text{else},
\end{array}\right.
\end{align*}
and $W_K\in\fspace{X}_{K}$ as the $2K$-periodic continuation of $V_K|_{I_K}$. This implies
\begin{align*}
\norm{V_K-W_\infty}_{2,\Rset}+\norm{\calA V_K-\calA W_\infty}_{r,\Rset}\quad\xrightarrow{K\to\infty}\quad0,\qquad
\end{align*}
for all $r\in\ccinterval{1}{\infty}$, and by continuity of $\calL_\infty$ and $\bar{\zeta}_\infty$ we get
\begin{align*}
\calL_\infty\at{W_\infty}=\lim_{K\to\infty}\calL_\infty\bat{V_K},\qquad
1=\bar\zeta_\infty\at{W_\infty}=\lim_{K\to\infty}\bar\zeta_\infty\bat{V_K}\,.
\end{align*}
Moreover, the properties of $\calA$ ensure 
\begin{align*}
\at{\calA{V}_K}\at{\phase}=\at{\calA{W}_K}\at{\phase}\quad\text{for all}\quad
\phase\in{I_K}=\ocinterval{-K}{K}\,,
\end{align*}
and thus we find
\begin{align*}
\calL_K\at{W_K}=\calL_\infty\bat{V_K},\qquad
\bar{\zeta}_{K}\bat{W_K}=\bar{\zeta}_\infty\at{V_K}.
\end{align*}
Consequently, we have $\calL_\infty\at{W_\infty}=\lim_{K\to\infty}\calL_K\at{\bar{\zeta}_K\at{W_K}{W}_K}$,
and $\ell_K\leq\calL_K\at{\bar{\zeta}_K\at{W_K}{W}_K}$ yields $\limsup_{K\to\infty}\ell_K\leq\calL_\infty\at{W_\infty}$. The thesis now follows since $W_\infty\in\nehari_\infty$ was arbitrary.
\end{proof}
\begin{corollary}%[Minimizers are non-constant]
\label{Cor:NonConstancy}
Let $K<\infty$ be sufficiently large. Then, each minimizer of $\calL_K|_{\nehari_K}$ is non-constant.
\end{corollary}
\begin{proof} 
The only constant function in $\nehari_K$ is $W\at\phase=\bar\zeta$ with 
$\bar\zeta\in\Rset$ being the unique maximizer $\bar\zeta:=\bar\xi\triple{\si^2/2}{1}{1}>0$ of the function
\begin{align*}
\xi\mapsto \frac{L_K\at{\xi}}{2K}=\la_{\triple{\si^2/2}{1}{1}}\at\xi=\tfrac{1}{2}\si^2\xi^2+\xi^q-\xi^p,
\end{align*}  
see Lemma \ref{Lem:AuxRes} and Corollary \ref{Cor:MaxAlongRays}. In particular, for all sufficiently large $K$ we have 
\begin{align*}
\calL_K\bat{\bar\zeta}=K\calL_{1}\at{\bar\zeta}>\ell_\infty\,,
\end{align*}
and hence 
$\calL_K\bat{\bar\zeta}>\ell_K$ due to Lemma \ref{Lem:LimSupEst}.
\end{proof}

Combing Theorem \ref{Thm:ExistencePeriodicWaves} and Lemma \ref{Lem:GradFlowProps} with Remark \ref{Rem:Nehari} and Corollary \ref {Cor:NonConstancy} we now obtain our existence result for periodic ground waves as formulated in Theorem \ref{Thm:MainResult}.
%
%
%-----------------------------------------------------------------------------------------------
\subsection{Convergence to solitary ground waves}\label{sec:min.sol}
%-----------------------------------------------------------------------------------------------
%
Our final goal is to prove that $\calL_\infty$ attains its minimum on $\nehari_\infty$. Since the operator
$\calA$ is no longer compact, we cannot argue as in the proof of Theorem \ref{Thm:ExistencePeriodicWaves}. Instead, 
we construct minimizers as limit of period ground waves. The same strategy was used in \cite[Section 3.4.2]{Pan05}
and some of our key arguments are similar to those presented there.
\begin{theorem} 
\notag%\label{Thm.Existence}
$\calL_\infty$ attains its minimum on $\nehari_\infty$ and we have
$\ell_\infty=\lim_{K\to\infty}\ell_K$. In particular, each unbounded sequence $\at{K_m}_{m\in\Nset}$ has at least one subsequence
$\at{K_n}_{n\in\Nset}$ with the following property: There exists a corresponding sequence $\at{W_n}_{n\in\Nset}$ of period ground waves $W_n\in\nehari_{K_n}$ that converges to a solitary ground wave $W_\infty\in\nehari_\infty$ in the sense of
\begin{align*}
\norm{W_\infty-V_n}_{2,\Rset}\quad\xrightarrow{n\to\infty}\quad0,
\end{align*}
where $V_n\in\fspace{X}_\infty$ is defined by $V_n\at{\phase}=W_n\at{\phase}$ for $\phase\in{I}_{K_n}$ and $
V_n\at{\phase}=0$ for $\phase\notin{I}_{K_n}$. 
\end{theorem}
\begin{proof} 
\underline{Step 1:} According to Theorem \ref{Thm:ExistencePeriodicWaves} and Remark \ref{Rem:Nehari}, for each $m$ there exists a periodic traveling wave $W_m\in{\nehari_{K_m}}$ which minimizes $\calL_{K_m}|_{\nehari_{K_m}}$. Since $\calL_K$ and traveling wave equation \eqref{Eqn:TW2} are invariant under shifts $W\rightsquigarrow{W}\at{\cdot+\phase_0}$, and 
since the function $\calA{W_m}$ is continuous, we can assume that 
\begin{align}
\label{Thm.Existence.Eqn2}
\at{\calA{W_m}}\at{0}=\norm{\calA W_m}_{\infty,I_{K_m}}.
\end{align}
By Lemma \ref{Lem:Nehari.Estimates} and Lemma \ref{Lem:LimSupEst} we also have
\begin{align}
\label{Thm.Existence.Eqn1}
{c}\leq\calL_{K_m}\at{W_{m}}\leq{C},\qquad
{c}\leq\norm{W_m}_{2,I_{K_m}}\leq{C}
\qquad{c}\leq\norm{\calA{W_m}}_{\infty,I_{K_m}}\leq{C}
\end{align}
and
\begin{align*}
\calQ_{K_m}\at{W_m}\leq{C}\,,\qquad \calP_{K_m}\at{W_m}\leq{C}
\end{align*}
for some constants $0<c<C<\infty$ independent of $m$. Moreover,  from \eqref{Eqn:TW2} we infer, using
Lemma \ref{Lem:AProps} and \eqref{Thm.Existence.Eqn1}, that
\begin{align*}
\norm{\at{\calA{W_m}}^\prime}_{\infty,I_{K_m}}\leq{2}
\norm{W_m}_{\infty,I_{K_m}}\leq{C}\norm{\calA W_m}_{\infty,I_{K_m}}\,.
\end{align*}
In view of \eqref{Thm.Existence.Eqn2} we thus obtain
\begin{align}
\label{Thm.Existence.Eqn3}
W_m\at{\phase} 
=% 
W_{m}\at{0}+\int_0^{\phase}\at{\calA{W}_m}^\prime\at{s}\dint{s}\geq{d}\qquad\text{for all}\;\; \abs{\phase}\leq{d}\,,
\end{align}
for some constant $d>0$ independent of $m$. 
\par%
\underline{Step 2:} By definition, we have $\norm{V_m}_{2,\Rset}=\norm{W_m}_{2,I_{K_m}}\leq{C}$ and
\begin{align}
\label{Thm.Existence.Eqn5}
\at{\calA{V}_m}\at{\phase}=\at{\calA{W}_m}\at{\phase}\qquad\text{ for all}\;\; \abs{\phase}<K_m-\tfrac12\,.
\end{align}
We now choose a subsequence $\at{K_n}_{n\in\Nset}$ such that
$V_n\rightharpoonup{W_\infty}\in\fspaceL^2\at\Rset$,
and Lemma \ref{Lem:AProps} provides the pointwise convergence
\begin{align}
\label{Thm.Existence.Eqn8}
\at{\calA{W}_\infty}\at{\phase}=\lim_{n\to\infty}\at{\calA{V_n}}\at{\phase}=
\lim_{n\to\infty}\at{\calA{W_n}}\at{\phase}\,.
\end{align}
Moreover, thanks to \eqref{Eqn:TW2} and \eqref{Thm.Existence.Eqn5} we conclude that the sequence $\at{V_n}_{n\in\Nset}$ converges pointwise, and combining this 
with $V_n\rightharpoonup{W_\infty}$ we arrive at the pointwise convergence
\begin{align}
\notag%\label{Thm.Existence.Eqn9}
W_\infty\at{\phase}=\lim_{n\to\infty}V_n\at{\phase}=\lim_{n\to\infty}W_n\at{\phase}
\end{align}
and we conclude that $W_\infty\in\fspaceL^2\at\Rset$ solves the traveling wave equation \eqref{Eqn:TW2}.
\par
\underline{Step 3:}  
For each $0<D<\infty$ and $1\leq{r}<\infty$ we have
\begin{align*}
\int_{-D}^D\abs{\calA W_\infty}^r\dint{\phase}=\lim_{n\to\infty}\int_{-D}^D\abs{\calA W_n}^r\dint{\phase}
\leq\liminf_{n\to\infty}\int_{I_{K_n}}\abs{\calA W_n}^r\dint{\phase},
\end{align*}
due to the pointwise convergence \eqref{Thm.Existence.Eqn8}, the estimate \eqref{Thm.Existence.Eqn1}, 
and the Dominated Convergence Theorem, and the limit $D\to\infty$ gives
\begin{align}
\label{Thm.Existence.Eqn6}
\calQ_\infty\at{W_\infty}\leq\liminf_{n\to\infty}\calQ_{K_n}\at{W_n}\leq{C}\,,\qquad
\calP_\infty\at{W_\infty}\leq\liminf_{n\to\infty}\calP_{K_n}\at{W_n}\leq{C}\,.
\end{align}
Similarly, using Fatou's Lemma and passing afterwards to $D\to\infty$ we prove that
\begin{align}
\label{Thm.Existence.Eqn7}
\norm{W_\infty}_{2,\Rset}^2=\lim_{D\to\infty}\int_{-D}^{D} W_\infty\at{\phase}^2\dint{\phase}\leq
\lim_{D\to\infty}\liminf_{n\to\infty}\int_{-D}^{D} W_n\at{\phase}^2\dint{\phase}\leq\liminf_{n\to\infty}\norm{W_n}_{2,I_{K_n}}^2.
\end{align}
We have now shown that $W_\infty$ is a traveling wave,
belongs to $\fspace{X}_\infty$, and does not vanish thanks to \eqref{Thm.Existence.Eqn3}.
By 
Remark \ref{Rem:KernelA} and Remark \ref{Rem:Nehari} we therefore find
$\calA{W}\neq0$ and hence $W_\infty\in\nehari_\infty$.
\par
\underline{Step 4:}
Combining \eqref{Thm.Existence.Eqn6} and \eqref{Thm.Existence.Eqn7} with $W_n\in{\nehari_{K_n}}$ for all $n\in\Nset\cup\{\infty\}$ we now estimate
\begin{align}
\ell_\infty\leq\calL_\infty\at{W_\infty}&=\si^2\at{\frac{1}{2}-\frac{1}{p}}\norm{W_\infty}^2_{2,\Rset}+
\at{1-\frac{q}{p}}\calQ_\infty\at{W_\infty}
\notag\\&\leq%
\label{Thm.Existence.Eqn4}
\si^2\at{\frac{1}{2}-\frac{1}{p}}\liminf_{n\to\infty}\norm{W_n}^2_{2,{I_{K_n}}}+
\at{1-\frac{q}{p}}\liminf_{n\to\infty}\calQ_n\at{W_n}
\\&=
\notag%
\liminf_{n\to\infty}\calL_{K_n}\at{W_n}=\liminf_{n\to\infty}\ell_{K_n}.
\end{align}
On the other hand, by Lemma \ref{Lem:LimSupEst} we have $\limsup_{n\to\infty}\ell_{K_n}\leq\ell_\infty$, and thus we find
\begin{align}
\notag%\label{Thm.Existence.Eqn3}
\ell_\infty=\calL_\infty\at{W_\infty}=\liminf_{n\to\infty}\ell_{K_n}=\limsup_{n\to\infty}\ell_{K_n}\,.
\end{align}
Consequently, $W_\infty$ is in fact a solitary ground wave and we have an equality sign in 
\eqref{Thm.Existence.Eqn4}. This  implies
\begin{align*}
\norm{W_\infty}_{2,\Rset}=\lim_{n\to\infty}\norm{W_n}_{2,{I_n}}=\lim_{n\to\infty}\norm{V_n}_{2,\Rset}\,,
\end{align*}
which in turn provides the strong convergence $V_n\to W_\infty$.  Finally, $\ell_\infty=\lim_{K\to\infty}\ell_{K}$ holds since we have shown that any unbounded sequence $\at{K_m}_{m\in\Nset}$ has a subsequence $\at{K_n}_{n\in\Nset}$ such that $\ell_\infty=\lim_{n\to\infty}\ell_{K_n}$.
\end{proof}
%
%
%
%-----------------------------------------------------------------------------------------------
\section{Numerical solutions}\label{sec:num}
%-----------------------------------------------------------------------------------------------
%
In order to illustrate our analytical findings we implemented the following
discretization of the constrained gradient flow for $\calL_K|_{\nehari_K}$:
\begin{enumerate}
\item
Given $0<K<\infty$, we divide the periodicity cell $\ocinterval{-K}{K}$ into $N$ equidistant grid points and
approximate all integrals by Riemann sums. 
\item
We choose a small flow time $\triangle\tau$ and minimize the action on $\nehari_K$ by iterating the following two update steps:
\begin{enumerate}
\item 
We compute an explicit Euler step for \eqref{Eqn:GradientFlow}, this means we update $W$ tangential to $\nehari_K$ via
\begin{align*}
W\mapsto {W}-\triangle\tau\Bat{\partial\calL_K\at{W}-\la_K\at{W}\partial\calF_K\at{W}}\,.
\end{align*}
\item 
We update in radial direction via 
\begin{align*}
W\mapsto \bat{1+\triangle\tau\calF_K\at{W}}W
\end{align*}
to enforce the constraint $\calF_K\at{W}=0$ (notice that $\bar\zeta_K\at{W}\lessgtr 1$ if and only if $\calF_K\at{W}\lessgtr0$). 
\end{enumerate}
\item To initialize the iteration, we discretize reasonable initial data such as
$W\at{\phase}=\exp\at{-\phase^2}$.
\end{enumerate}
\begin{figure}[t!]
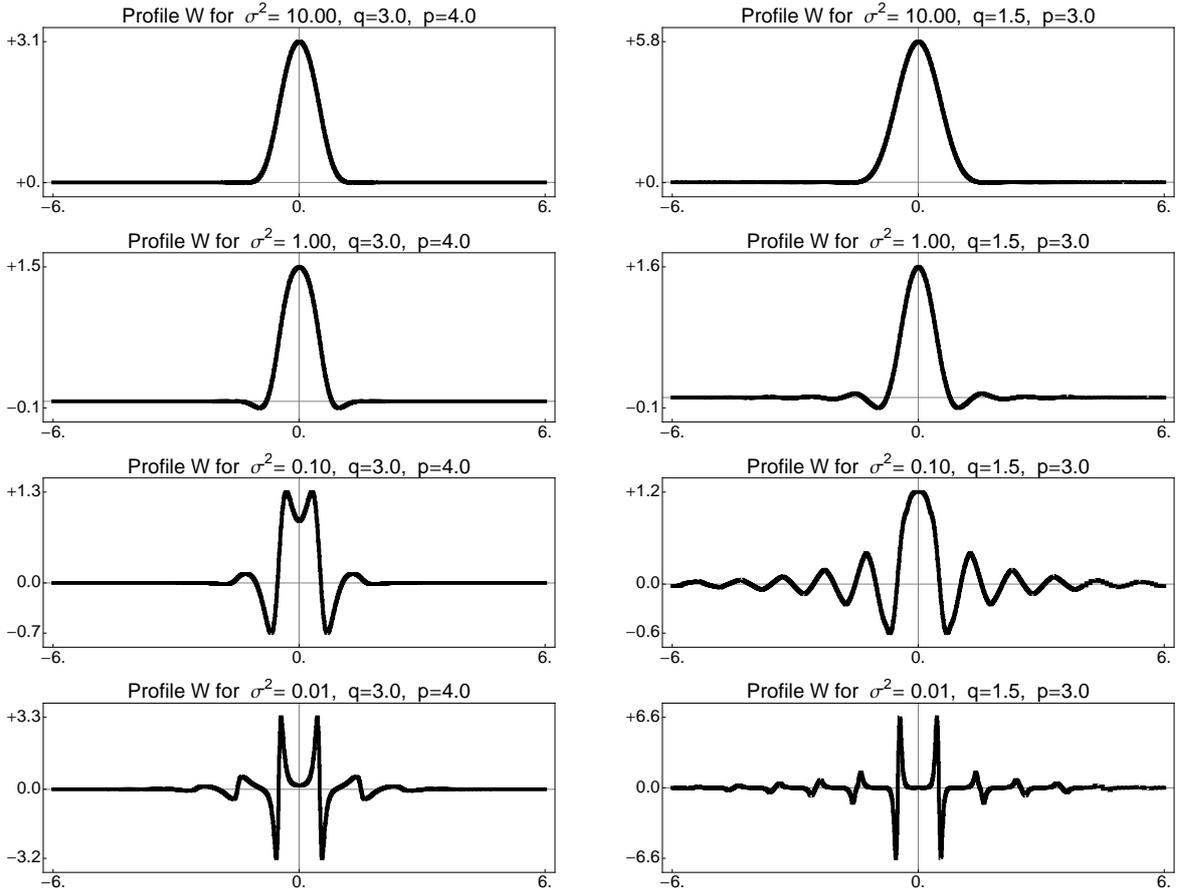
%
\centering{%
\includegraphics[width=0.45\textwidth, draft=\figdraft]{\figfile{prof_24}}%
\hspace{0.05\textwidth}%
\includegraphics[width=0.45\textwidth, draft=\figdraft]{\figfile{prof_14}}%
\\%
\includegraphics[width=0.45\textwidth, draft=\figdraft]{\figfile{prof_23}}%
\hspace{0.05\textwidth}%
\includegraphics[width=0.45\textwidth, draft=\figdraft]{\figfile{prof_13}}%
\\%
\includegraphics[width=0.45\textwidth, draft=\figdraft]{\figfile{prof_22}}%
\hspace{0.05\textwidth}%
\includegraphics[width=0.45\textwidth, draft=\figdraft]{\figfile{prof_12}}%
\\%
\includegraphics[width=0.45\textwidth, draft=\figdraft]{\figfile{prof_21}}%
\hspace{0.05\textwidth}%
\includegraphics[width=0.45\textwidth, draft=\figdraft]{\figfile{prof_11}}%
}%
\caption{%
Numerical approximations of ground wave for different values of $\si^2$ and
$q=3,p=4$ (left column) or $q=3/2,p=3$ (right column). The numerical parameters are $K=6$, $N=2400$,
and $\triangle\tau\geq0.0005$.
}%
\label{Fig}%
\end{figure}%
Although this scheme is rather simple, it exhibits good convergence properties provided that the flow time $\triangle\tau$ is sufficiently small. Moreover, numerical simulations indicate that each time-discrete trajectory converges to a limit that is independent of the particular choice of the initial data,
and thus we expect that this limit approximates a global minimizer of $\calL_K|_{\nehari_K}$. 
\par
Typical numerical solutions for $K=6$ and $\si\in\{0.01,0.1,1.0,10.0\}$ are displayed in Figure \ref{Fig}. We clearly observe that
periodic ground waves are localized but have rather different shapes for large and small speeds, respectively. If $\si^2$ is sufficiently large, the double-well structure of $\Phi$ is less important and the wave looks like a unimodal 
wave for the convex potential $\Phi\at{w}=w^p$, see \cite{Her10a} for details. For small $\si^2$, however, the contributions from the different monomials balance and $W$ has several local extrema.
%
%
%
%
%-------------------------------------------------------------------------------------------
%                      Literature
%-------------------------------------------------------------------------------------------
%
%%%\bibliographystyle{amsalpha}%
%%%\bibliography{mima}%

\providecommand{\bysame}{\leavevmode\hbox to3em{\hrulefill}\thinspace}
\providecommand{\MR}{\relax\ifhmode\unskip\space\fi MR }
\providecommand{\MRhref}[2]{%
  \href{http://www.ams.org/mathscinet-getitem?mr=#1}{#2}
}
\providecommand{\href}[2]{#2}

\end{document}